\documentclass[11pt]{article}
\usepackage{amsmath, amssymb, graphicx, enumitem}
\usepackage{amsthm}
\usepackage[margin=1in]{geometry}
\usepackage{url}
\usepackage{graphicx}
\usepackage{amsfonts}
\usepackage{amssymb}
\usepackage{amsmath}
\usepackage{latexsym}
\usepackage{hyperref}
\usepackage{xcolor}
\usepackage{cleveref}
\usepackage{multirow}
\usepackage{bm}
\usepackage{float}
\usepackage[normalem]{ulem}
\usepackage{dblfloatfix}
\usepackage{cite}
\usepackage{caption}

\usepackage[ruled,vlined]{algorithm2e}
\SetKwInOut{Input}{Input}
\SetKwInOut{Output}{Output}
\SetKw{Continue}{continue}
\SetCommentSty{mycommentfont}

\newcommand{\ABox}{
\raisebox{3pt}{\framebox[6pt]{\rule{6pt}{0pt}}}
}
\newtheorem{theorem}{Theorem}
\newtheorem{proposition}{Proposition}
\newtheorem{remark}{Remark}
\newtheorem{definition}{Definition}

\crefname{theorem}{Theorem}{Theorems}
\crefname{lemma}{Lemma}{Lemmas}
\crefname{proposition}{Proposition}{Propositions}
\crefname{definition}{Definition}{Definitions}
\crefname{corollary}{Corollary}{Corollaries}
\crefname{figure}{Figure}{Figures}
\crefname{table}{Table}{Tables}
\crefname{property}{Property}{Properties}
\crefname{remark}{Remark}{Remarks}
\crefname{section}{Section}{Sections}
\crefname{appendix}{Appendix}{Appendices}

\newtheorem{lem}{Lemma}
\crefname{lem}{Lemma}{Lemmas}

\newcommand{\squeezelist}{\setlength{\itemsep}{0pt}}

\title{Pinch-Free Orientations of Polyomino Grid Graphs}

\author{James Fox\thanks{Department of Computer Science, Villanova University, Villanova, PA, USA, \texttt{jfox26@villanova.edu}}
	\and
Mirela Damian\thanks{Department of Computer Science, Villanova University, Villanova, PA, USA, \texttt{mirela.damian@villanova.edu}}}

\date{}

\begin{document}
\maketitle

\begin{abstract}
We study lattice-based collective construction of a target polyomino on the integer grid, where robots place unit square blocks subject to two physical constraints: the partial structure must remain connected at all times, and no placement may create an unfillable one-cell gap. We present a wavefront sweep algorithm that computes a dependency structure satisfying the local conditions required for safe assembly. Given a list of $n$ cells sorted lexicographically by coordinates, the orientation can be computed in $O(n)$ time using $O(n)$ auxiliary space.
\end{abstract}

\section{Introduction}
Collective robotic construction builds large structures by having many relatively simple,
autonomous robots transport and place passive building blocks, coordinating through local
sensing, limited communication, and decentralized decision-making. By exploiting
parallelism and avoiding a single point of control, such systems are scalable, robust to
individual failures, and well suited to uncertain, changing environments~\cite{Petersen2019Review}.

In this paper we study \emph{lattice-based collective construction}, where robots assemble
structures on a two-dimensional grid using unit square blocks. This abstraction captures
key geometric and coordination constraints of block-based building while enabling clean
algorithmic formulations and provable guarantees. A central difficulty is ensuring that
\emph{every intermediate configuration} remains physically buildable. We focus on two
fundamental physical constraints:
(i) the partially built structure must remain connected at all times; and
(ii) a robot cannot place a block into a one-cell gap between two already placed blocks,
because such a gap is too narrow to maneuver a block into position.
A poor attachment sequence can therefore create narrow, inaccessible sites that halt progress.

These constraints motivate construction plans that can be enforced using only local
information while still guaranteeing a feasible build sequence. We model the target as a
polyomino $P$ and represent it by its dual graph $G$, with one vertex per occupied cell
and an edge between adjacent cells. Our goal is to compute a dependency structure on $G$
that encodes safe attachment conditions for decentralized construction while preserving
connectivity and avoiding one-cell gaps throughout the build.

Fitzsimmons and Flatland~\cite{FitzsimmonsFlatland2011} showed that these constraints can
be captured by an orientation that is rooted, acyclic, and \emph{pinch-free}, meaning that
no vertex has collinear incoming edges. We call
such an orientation a \emph{rooted pinch-free orientation}. 
In the unrestricted-access model, where robots can access any empty grid cell regardless of 
%regardless of whether it is reachable from the unbounded free region 
reachability from the unbounded region, any topological order of such an orientation yields a valid local construction plan.

\paragraph{Our contribution.}
We give a sweep-based wavefront algorithm that computes a rooted pinch-free orientation
for any connected polyomino, including polyominoes with holes. If the input is given as a
list of $n$ cells sorted lexicographically by coordinates, the algorithm runs in $O(n)$ time
and uses $O(n)$ auxiliary space. For an unsorted list, the preprocessing time is
$O(n\log n)$ by comparison sorting, or $O(n)$ by radix sort when the integer coordinates are in a bounded range. This improves the previously implicit $O(n^2)$ bound
from~\cite{FitzsimmonsFlatland2011}.

%\medskip
\noindent
\begin{minipage}{0.6\textwidth}
This output can be viewed as a constrained graph orientation, but the requirement that no two incoming edges be collinear in the underlying grid embedding is not captured by standard graph-orientation primitives such as $st$-orientations~\cite{EvenTarjan1976}, nor is it enforced by BFS-type traversals: in~\cref{fig:bfs}, a BFS-induced orientation from $s$ creates a pinch at $t$. 
\end{minipage}\hfill
\begin{minipage}{0.3\textwidth}
% \vspace{-0.2em}
\centering
\includegraphics[page=4,width=0.45\linewidth]{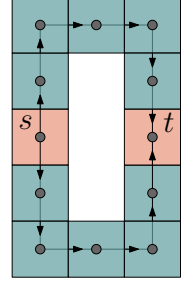}
\captionsetup{
        font=scriptsize,
        justification=raggedright,
        singlelinecheck=false
}
 \captionof{figure}{A BFS-induced pinch.}
\label{fig:bfs}
\end{minipage}

%\red{This output can be viewed as a constrained graph orientation, but its nonstandard feature is grid-geometric: besides being acyclic and rooted, the orientation must forbid two incoming edges that are collinear in the underlying grid embedding. This condition is not captured by standard graph-orientation primitives such as $st$-orientations~\cite{EvenTarjan1976}. Similarly, classical traversals such as BFS do not necessarily induce a pinch-free orientation; in the example of~\cref{fig:poly}(c), the BFS order leaves a cross edge connecting the two BFS branches, and either orientation of this edge creates a pinch.} 

\vspace{-1em}
\subsection{Related work}
Collective robotic construction spans many systems and coordination strategies; see
Petersen et al.~\cite{Petersen2019Review} for a survey. Many such systems achieve
scalability and robustness by avoiding global state and coordinating through locally
observable cues in the shared environment, a stigmergic strategy inspired by ants and
termites. Termite-inspired construction is exemplified by the TERMES line of
work~\cite{Petersen2011Termes,Werfel2014Science}.

Decentralized robotic construction is closely related to lattice-based
self-reconfiguring modular robotics, where identical units transform one shape into
another through local module motions 
such as pivoting~\cite{MTRAN13, Rus15,Musketeers21,pivoting-socg21,Halil25}, sliding~\cite{Dumitrescu04,Dumitrescu06,VJ20,Hugo22,ZA+24,Irina24,Brunner2024,Hugo25input},
or expansion~\cite{GS+08,Crystalline09}. 
By contrast, construction models build a 
target shape incrementally, adding modules to a growing structure. 

Closest to our work are grid-based construction models involving cubic modules, in which the main geometric constraints are sufficient clearance to move through free space, and access to each placement site from the unbounded exterior. 
Werfel and Nagpal~\cite{Werfel2006} give a perimeter-following strategy for connected 
2D targets without holes and with sufficiently wide corridors, and Werfel et al.~\cite{Werfel2008} 
extend this framework to a restricted class of 3D targets using local attachment rules 
encoded in the blocks. 

The MIT--NASA Space Robots Team~\cite{Brunner2024} studies construction and reconfiguration in a 3D motion model for cubic modules, where moving modules must travel through free space with sufficient clearance. Their algorithms deconstruct extremal slices within a broader 3D motion-planning framework, preserving exterior accessibility and the required clearance properties under their assumptions. Our work isolates the intrinsic 2D dependency structure of a target slice: for any connected 2D target, we compute in linear time a rooted pinch-free orientation, equivalently a partial order for construction or deconstruction, whose prefixes preserve connectivity within the 2D target. 

Fitzsimmons and Flatland~\cite{FitzsimmonsFlatland2011} study the corresponding 2D
construction problem and give an $O(n^2)$ algorithm based on reverse peeling of non-cut
vertices. Their construction sequence induces an acyclic dependency structure,
equivalently a rooted pinch-free orientation. 
The quadratic bound comes from a linear-time search for the next valid peel at each step; speeding up cut-vertex maintenance alone does not directly make this geometrically constrained peeling procedure linear. We avoid this dynamic search by making the induced structure itself the object of computation and computing it directly in linear time by a wavefront sweep. 

A different line of work prescribes a single global traversal: Damian and
Flatland~\cite{DF15} use squeeze-free Hamiltonian paths to avoid inaccessible one-cell
gaps, trading flexibility for minimal coordination.

\section{Problem Formulation}
\label{sec:problem}

A \emph{polyomino} is a finite connected union of unit grid cells joined edge to edge.
Let $G$ be its dual graph: each occupied cell is a vertex, and adjacent cells are joined
by an edge; see~\cref{fig:poly}~for an example. Thus $G$ is a finite connected subgraph of the integer grid, with maximum degree~4.

%A \emph{polyomino} is a finite connected union of unit grid cells joined edge to edge. Given a polyomino $P$, let $G$ denote its dual graph, obtained by placing one vertex in each occupied cell and joining two vertices by an edge whenever the corresponding cells share a side; see~\cref{fig:poly}~for an example. Thus $G$ is a finite, simple, connected subgraph of the integer grid, and every vertex has degree at most~4, since each cell has at most four cardinal neighbors.

%\begin{definition}[Pinch at a vertex]
%\label{def:pinch}
%Let $\overrightarrow{G}$ be an orientation of $G$. A vertex $v$ of $G$ is
%\emph{pinched} in $\overrightarrow{G}$ if it has two collinear incoming edges;
%otherwise $v$ is \emph{pinch-free}.
%\end{definition}

Let $\overrightarrow{G}$ be an orientation of $G$. A vertex $v$ of $G$ is
\emph{pinched} in $\overrightarrow{G}$ if it has two collinear incoming edges. 

\begin{definition}[Rooted pinch-free orientation]
\label{def:roo}
An orientation $\overrightarrow{G}$ of $G$ is called \emph{rooted pinch-free} if it
satisfies the following properties: 
\begin{description}
\squeezelist
\item[(P1) Pinch-free:] no vertex is pinched;
\item[(P2) Acyclic:]  $\overrightarrow{G}$ is acyclic; 
\item[(P3) Rooted:] some root vertex $r$ reaches every other vertex of $G$ by a directed path in $\overrightarrow{G}$.
\end{description}
\end{definition}

\noindent
Fitzsimmons and Flatland~\cite{FitzsimmonsFlatland2011} showed that attaching blocks in any topological order of a rooted pinch-free orientation keeps the partial
construction connected and pinch-free at every step. Our main result is the following. 

\begin{theorem}[Main result]
\label{thm:main}
Let $P$ be a connected polyomino with $n$ cells, given as a list of cell coordinates
sorted lexicographically by $(x,y)$, and let $G$ be its dual graph. Then a rooted
pinch-free orientation of $G$ can be computed in $O(n)$ time using $O(n)$ auxiliary space. 
\end{theorem}

\noindent
The proof is organized as follows. \cref{sec:wavefront} gives the algorithm,
\cref{sec:correctness} proves correctness,
and \cref{sec:complexity} analyzes the time and space bounds.

%\section{Constructing a Rooted Pinch-Free Orientation}
\section{Wavefront Orientation}
\label{sec:algorithm}

We describe an algorithm that constructs a rooted pinch-free orientation of the dual
graph $G$ of a polyomino $P$.

\begin{definition}[Column components]
\label{def:col-decomp}
A \emph{column component} of $G$ is a maximal vertically connected subset of the vertices in a single column; the collection of all such components is the \emph{column decomposition} 
of $G$ (see~\cref{fig:poly}).
\end{definition}

Two distinct column components $H$ and $K$ are adjacent if $G$ has a
horizontal edge with one endpoint in $H$ and the other in $K$. In~\cref{fig:poly} 
for example, $H$ and $K$ are adjacent, but $H$ and $J$ are not. 
Let $x(H)$ denote the column index of component $H$; if
$H$ and $K$ are adjacent, then $x(K)=x(H)\pm1$.

%For a column component $H$, let $x(H)$ denote its column index.
%Thus, if $H$ and $K$ are adjacent, then $x(K)=x(H)\pm 1$.

\begin{figure}[!htp]
    \centering
    \includegraphics[page=1,width=0.38\textwidth]{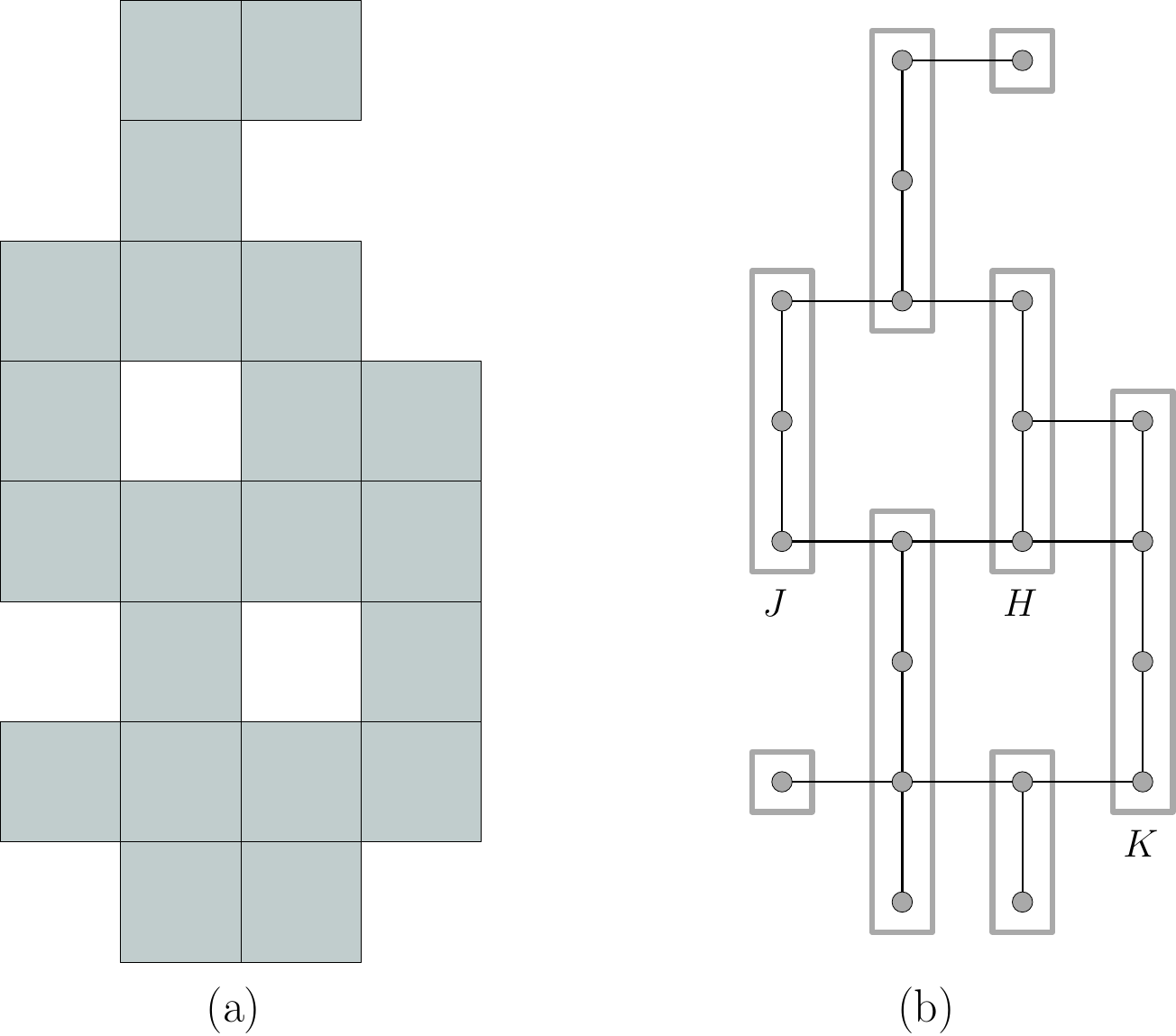}
\caption{(a) A polyomino example; (b) its dual graph and column decomposition.}    
\label{fig:poly}
\end{figure}

\begin{definition}[Forward and inward neighbors]
\label{def:forward-inward}
Fix a horizontal direction $d\in\{-1,+1\}$, where $d=+1$ means right and
$d=-1$ means left. Let $H$ and $K$ be adjacent column components.
We call $K$ a \emph{forward neighbor} of $H$ with respect to $d$ if
$x(K)=x(H)+d$, and an \emph{inward neighbor} of $H$ with respect to $d$ if
$x(K)=x(H)-d$.
The \emph{forward side} of $H$ is the
side of column $x(H)$ in direction $d$; the \emph{inward side} is the opposite side.
%The \emph{forward side} of $H$ is the open halfplane, bounded by the vertical line through column $x(H)$, that contains the forward neighbors of $H$. The \emph{inward side} of $H$ is defined analogously.
\end{definition}

\subsection{Wavefront Sweep Algorithm}\label{sec:wavefront}

We now describe \emph{Wavefront}, our algorithm for orienting $G$. It performs a sequence of horizontal sweeps, or \emph{waves}, alternating between left-to-right and right-to-left directions. Each wave has a fixed direction $d\in{-1,+1}$ and grows the visited set of column components from a current \emph{frontier}. The first frontier consists of a single chosen root component; each later frontier consists of the unvisited components exposed on the inward side of the previous wave.

\vspace{-0.8em}
\paragraph{Preprocessing and initialization.}
We first compute the column components and their horizontal adjacencies.
Choose an arbitrary root component $H_{\mathrm{root}}$. For each component $H$, mark $H$
unvisited and initialize its seed $s(H)\gets\bot$, except for the root component, whose
seed is set to an arbitrary vertex of $H_{\mathrm{root}}$. We also maintain two global
frontiers $\Gamma_{-1}$ and $\Gamma_{+1}$, initially empty except for
$\Gamma_{+1}=\{H_{\mathrm{root}}\}$. Thus the first wave runs with direction $d=+1$.

\vspace{-0.5em}
\paragraph{A single wave.}
A call \textsc{Wave}$(d)$ launches a wave in direction $d$ from the global frontier
$\Gamma_d$. It first transfers $\Gamma_d$ to a local worklist $W$ and clears
$\Gamma_d$. It then repeatedly removes a component $H$ from $W$; visited components
are skipped, and each unvisited component is processed as follows:

%\begin{enumerate}[leftmargin=*, itemsep=2pt, topsep=2pt]
\begin{enumerate}[label=\textbf{\arabic*.}, leftmargin=*, itemsep=1pt, topsep=2pt]
\item \textbf{Choose seed.}
If $s(H)=\bot$, choose $s(H)$ adjacent to a visited component
in column $x(H)-d$.

\item \textbf{Orient vertical edges.}
Orient all vertical edges of $H$ away from $s(H)$.

\item \textbf{Expand forward.}
For each forward neighbor $K$ of $H$, orient all horizontal edges between $H$ and $K$
from $H$ into $K$. If $K$ is unvisited and not already in $W$, insert it into $W$.

\item \textbf{Build the next frontier.}
For each inward neighbor $K$ of $H$, if $K$ is visited and the horizontal edges between
$K$ and $H$ are still unoriented, orient them from $K$ into $H$. If $K$ is unvisited and
not already in $\Gamma_{-d}$, insert it into $\Gamma_{-d}$.

\item \textbf{Finalize.}
Mark $H$ as visited.
\end{enumerate}

\noindent
The wave ends when $W$ is empty. At that point, $\Gamma_{-d}$ contains the starting
frontier for the next wave. The processing order inside $W$ is arbitrary; correctness
does not depend on the order in which components are removed from the local worklist $W$.

%\begin{figure}[!htp]
\begin{figure*}[t]
    \centering
    \includegraphics[page=2,width=0.7\textwidth]{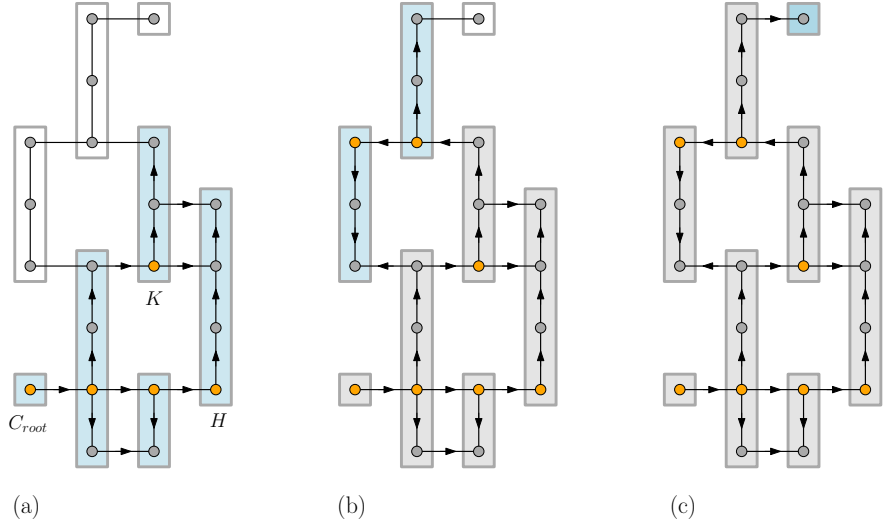}
    \caption{Wavefront orientation algorithm. In each panel, column components processed in the current wave are shown in blue;
previously processed components are shown in gray; seeds are shown in orange. The orientation $\overrightarrow{G}$
after the (a) first wave ($d=+1$), (b) second wave ($d=-1$), and (c) third wave ($d=+1$).}
      \label{fig:alg}
\end{figure*}

Figure~\ref{fig:alg} illustrates this process. In each wave, vertical edges are
oriented away from the seed of each processed component, while newly oriented horizontal
edges follow the wave direction. Inward neighbors encountered during one wave form the
frontier for the next. A component may also remain as a stale entry in the opposite
frontier after becoming visited; if so, it is simply skipped when later removed. 
For example, in~\cref{fig:alg}(a), if $H$ is processed before $K$ during a wave with direction $+1$, then $K$ is added to
$\Gamma_{-1}$ as an inward neighbor of $H$. Later in the same wave, $K$ is visited.
Thus, when $K$ is removed from $\Gamma_{-1}$ in the next wave, the algorithm detects
that it is already visited and skips it. Hence each component is processed at most once.

\begin{algorithm*}[htpb]
\DontPrintSemicolon
\caption{Wavefront Orientation}
\label{alg:wavefront}
\Input{The dual graph $G$ of a connected polyomino.}
\Output{A rooted pinch-free orientation $\overrightarrow{G}$ of $G$.}

Compute the column components of $G$ and their horizontal adjacencies.\;
For each column component $H$, mark $H$ unvisited and set $s(H)\gets\bot$.\tcp*{no seed yet}
Initialize the global frontier worklists $\Gamma_{-1}$ and $\Gamma_{+1}$ to empty.\;

Set $H_{\mathrm{root}}\gets$ an arbitrary column component.\;
Set $s(H_{\mathrm{root}})\gets$ an arbitrary vertex of $H_{\mathrm{root}}$.\tcp*{root seed}
Set $d\gets +1$ and $\Gamma_d\gets\{H_{\mathrm{root}}\}$.\tcp*{first wave: left-to-right}

\While{$\Gamma_d\neq\emptyset$}{
  \textsc{Wave}$(d)$.\tcp*{consumes $\Gamma_d$, fills $\Gamma_{-d}$}
  $d\gets -d$.\tcp*{alternate direction}
}

\vspace{0.6em}
\hrule
\vspace{0.6em}

\textsc{Wave}$(d)$\;
\Input{A direction $d\in\{-1,+1\}$; $\Gamma_d$ is the starting frontier for this wave.}
\Output{$\Gamma_d$ is exhausted, and $\Gamma_{-d}$ is populated for the next wave.}

Let $W\gets \Gamma_d$ and set $\Gamma_d\gets\emptyset$.\tcp*{move global frontier into local worklist}

\While{$W\neq\emptyset$}{
  Remove some component $H$ from $W$.\;
  \If{$H$ is visited}{
    continue\;
  }
  \If{$s(H)=\bot$}{
    Set $s(H)$ to any vertex of $H$ with a visited neighbor in column $x(H)-d$.\;
  }

  Orient all vertical edges of $H$ away from $s(H)$.\;

\vspace{1em}
  \tcc{Forward side: expand this wave in direction $d$.}
  \ForEach{forward neighbor component $K$ of $H$}{
    Orient all horizontal edges between $H$ and $K$ from $H$ into $K$.\;
    \If{$K$ is unvisited and $K\notin W$}{
      Insert $K$ into $W$.\;
    }
  }

\vspace{1em}
  \tcc{Inward side: build the frontier for the next wave in direction $-d$.}
  \ForEach{inward neighbor component $K$ of $H$}{
    \If{$K$ is visited and no horizontal edge between $H$ and $K$ is oriented yet}{
      Orient all horizontal edges between $K$ and $H$ from $K$ into $H$.\;
    }
    \If{$K$ is unvisited and $K\notin \Gamma_{-d}$}{
      Insert $K$ into $\Gamma_{-d}$.\;
    }
  }
  Mark $H$ as visited.\;
}
\hrule
\vspace{0.3em}
\centerline{\small{An interactive visual demo is available at 
% \url{https://papaya-semifreddo-2ee7ad.netlify.app/}.}}
%\url{https://anoncode10.github.io/robotic-assembly-demo/}.}}
\url{https://enabledornot.github.io/robotic-grid-assembly-demo/}.}}
\end{algorithm*}

\vspace{-0.5em}
\paragraph{Multiple waves and global control.}
Algorithm~\ref{alg:wavefront} repeatedly calls \textsc{Wave}$(d)$, starting with
$d=+1$ and then alternating $d\gets -d$. Each call consumes $\Gamma_d$ and may
populate only $\Gamma_{-d}$. The algorithm continues while the frontier for the current direction is nonempty, and terminates when both $\Gamma_{+1}$ and
$\Gamma_{-1}$ are empty.
We next prove that Wavefront  produces a rooted pinch-free orientation. 

%Algorithm~\ref{alg:wavefront} separates the global control from the local sweep: the top-level controller initializes the data and alternates the direction, while \textsc{Wave}$(d)$ executes one sweep.

%The correctness proof does not depend on the order in which components are removed from the local worklist $W$. 
%for \emph{any} such processing order.

\subsection{Correctness of the Wavefront Orientation}
\label{sec:correctness}

Let $\overrightarrow{G}$ be the orientation produced by Algorithm~\ref{alg:wavefront}.
We use the following notation. For each column component $H$, let $w(H)$ be the
index of the wave in which $H$ is processed, and let $d(H)\in\{-1,+1\}$ be the
direction of that wave. For each vertex $v$ of $G$, let $C_v$ denote the column
component containing $v$.

Whenever a component $H$ is processed in direction $d$, the algorithm applies
three rules: 
\begin{itemize}
\squeezelist
\item[(V)] all vertical edges in $H$ are oriented away from its seed $s(H)$;
\item[(F)] for each forward neighbor component $K$ of $H$, all horizontal edges
between $H$ and $K$ are oriented from $H$ into $K$;
\item[(I)] for each inward neighbor component $K$ of $H$, if $K$ is already visited
and no horizontal edge between $H$ and $K$ has yet been oriented, all such edges are
oriented from $K$ into $H$.
\end{itemize}

%Whenever a component $H$ is processed in direction $d$, the algorithm applies
%three rules: \textup{(V)} all vertical edges of $H$ are oriented away from its
%seed $s(H)$; \textup{(F)} all horizontal edges from $H$ to a forward neighbor are
%oriented from $H$ into that neighbor; and \textup{(I)} any still-unoriented
%horizontal edges from a visited inward neighbor are oriented into $H$. In
%particular, all horizontal edges oriented during a call to \textsc{Wave}$(d)$
%point consistently in direction $d$.

\newpage
\begin{remark}[Order of processing]
The worklists in Algorithm~\ref{alg:wavefront} need not enforce any particular order.
The correctness proof in this section shows that Wavefront produces a
rooted pinch-free orientation for \emph{any} processing order. 
One could impose additional geometric structure within a wave, for example by processing the components in increasing $x$-coordinate when $d=+1$ and in decreasing $x$-coordinate when $d=-1$, but no such structure is used in the correctness proof.
\end{remark}

\medskip
\noindent\textbf{Structural properties of the execution.}
We begin with two basic properties of the execution: every column component is processed
exactly once, and every edge of $G$ is oriented exactly once.

\begin{proposition}%[Processing and occurrences in $W$]
\label{prop:W-bound}
Throughout the execution, each column component is processed exactly once.
Moreover, each column component appears in a local worklist $W$ exactly once while
unvisited, and at most once more later while visited.
\end{proposition}

\begin{proof}
A component is processed only when it is removed from $W$ and found to be unvisited;
afterward it is marked visited. Hence no component can be processed more than once.

We next show that every component is processed at least once. Suppose not, and let $K$ be
an unprocessed component adjacent to a visited component $A$; such a pair exists because the polyomino is connected. Let $d=d(A)$ be the direction of the wave in which $A$ is
processed. Then, with respect to $d$, the component $K$ is either a forward neighbor or an
inward neighbor of $A$.

If $K$ is a forward neighbor of $A$, then during the processing of $A$, rule~\textup{(F)}
inserts $K$ into the current worklist $W$, if not already present. Thus $K$ is processed later in the same wave.

If $K$ is an inward neighbor of $A$, then processing $A$ inserts $K$ into the opposite frontier $\Gamma_{-d}$, if not already present. It may also happen that $K$ is encountered as a forward neighbor of another component in the same wave, in which case $K$ is added to the current local worklist $W$. In this case, $K$ is processed in the current wave; otherwise, as an element of $\Gamma_{-d}$, it is processed in the next wave. Even if $K$ is inserted in both $K_{-d}$ and $W$, it is processed only when first removed while still unvisited; any later removal is ignored because $K$ is found visited. In all cases, $K$ is eventually processed, a contradiction. Therefore every component is processed at least once, and hence exactly once.

Since a component is processed only when it is removed from $W$ and found to be unvisited,
each component must appear in some local worklist $W$ unvisited at least once. After it is
processed, it is marked visited and can never become unvisited again. Thus each component
appears in a local worklist $W$ exactly once while unvisited.
Once a component $K$ is marked visited, it is never inserted into a frontier again.
Thus $K$ can appear later in a local worklist $W$ only as a stale frontier entry,
inserted into some $\Gamma_{-d}$ while still unvisited and transferred to $W$ at the
start of the next call to \textsc{Wave}$(-d)$. This can happen at most once.
Hence each component appears in a local worklist $W$ at most once while visited.
\end{proof}

\begin{lem}
\label{lem:oriented-once}
Every edge of $G$ is oriented exactly once in $\overrightarrow{G}$.
\end{lem}

\begin{proof}
We consider vertical and horizontal edges separately.

\smallskip\noindent
\emph{Vertical edges.}
A vertical edge lies entirely within a unique column component $H$. By rule~\textup{(V)},
all vertical edges of $H$ are oriented when $H$ is processed, and no other rule
orients vertical edges. Since each component is processed exactly once by
\cref{prop:W-bound}, every vertical edge is oriented exactly once.

\medskip
\noindent\emph{Horizontal edges.}
Let $e$ be a horizontal edge between adjacent column components $A$ and $B$, and
suppose $A$ is processed before $B$. Let $d=d(A)$ be the direction of the wave
that processes $A$. Since $B$ is still unvisited at that moment, $e$ is still
unoriented when $A$ is processed.

If $B$ is a forward neighbor of $A$ with respect to $d$, then rule~\textup{(F)}
orients all horizontal edges between $A$ and $B$, including $e$, from $A$ into $B$.

Otherwise, $B$ is an inward neighbor of $A$. Then rule~\textup{(I)} does not apply
while processing $A$, because $B$ is unvisited. Thus $e$ remains unoriented after
$A$ is processed, and $B$ is inserted into the opposite frontier $\Gamma_{-d}$.
When $B$ is later processed, the edge $e$ is still unoriented and $A$ is already visited. 

If $B$ is processed later in the same wave, then $A$ is a forward neighbor of $B$ with respect to $d$, and rule~\textup{(F)} orients $e$ from $B$ into $A$.
Otherwise, $B$ is processed in the subsequent wave, whose direction is $-d$; then $A$ is an inward neighbor of $B$, and
rule~\textup{(I)} orients $e$ from $A$ into $B$. 

Thus every horizontal edge is oriented at least once. It is oriented at most once
because only rules~\textup{(F)} and~\textup{(I)} can orient horizontal edges between
$A$ and $B$, each component is processed exactly once by \cref{prop:W-bound}, and
rule~\textup{(I)} applies only when those edges are still unoriented.
Therefore every edge of $G$ is oriented exactly once.
\end{proof}

\noindent
We next establish a monotone potential function that will imply acyclicity.

\begin{definition}[Potential function]
\label{def:potential}
For a vertex $v$, let $\delta(v)$ be the length of the unique vertical directed path
from $s(C_v)$ to $v$ in $\overrightarrow{G}$.  
The \emph{potential} of $v$ is $
\Phi(v)\;=\;\bigl(w(C_v),\; d(C_v)\cdot x(v),\; \delta(v)\bigr)$, 
ordered lexicographically.
\end{definition}

Since each column component is a vertical path, rule~\textup{(V)} orients its
vertical edges away from the seed, so the quantity $\delta(v)$ is well defined. 

\begin{lem}\label{lem:potential}
For every directed edge $u\to v$ of $\overrightarrow{G}$, we have $\Phi(u)<\Phi(v)$.
More precisely:
\begin{enumerate}[label=\textup{(\roman*)}, leftmargin=*, itemsep=0pt, topsep=3pt]
\item if $u\to v$ is vertical, then the first
two coordinates of $\Phi$ are equal, and the third strictly increases;
\item if $u\to v$ is horizontal, then the third coordinate plays no role, and the potential increase  is determined by the first two coordinates.
\end{enumerate}
\end{lem}

\begin{proof}
Suppose first that $u\to v$ is vertical. Then $C_u=C_v$, and rule~\textup{(V)} orients
all vertical edges of $C_u$ away from $s(C_u)$. Hence
$\delta(v)=\delta(u)+1$, while the first two coordinates of $\Phi$ are unchanged.
Therefore $\Phi(u)<\Phi(v)$, proving~(i).

Now suppose that $u\to v$ is horizontal. Then $C_u$ and $C_v$ are distinct adjacent
column components.

\smallskip\noindent
\emph{Case \textup{(I)}.}
If $u\to v$ is oriented by rule~\textup{(I)} while processing $C_v$, then $C_u$ is
already visited and lies on the inward side of $C_v$. 
It cannot have been processed in the same wave as $C_v$, since then 
$C_v$ would have been a forward neighbor of $C_u$ and rule~\textup{(F)} would already 
have oriented all edges between them. Thus $w(C_u)<w(C_v)$, 
so the first coordinate of $\Phi$ strictly increases. Hence
$\Phi(u)<\Phi(v)$.

\smallskip\noindent
\emph{Case \textup{(F)}.}
Suppose instead that $u\to v$ is oriented by rule~\textup{(F)} while processing $C_u$
in direction $d$, with $C_v$ a forward neighbor of $C_u$.

If $C_v$ is unvisited at that time, then $C_v$ is inserted into the current local
worklist $W$ and is processed in the same wave as $C_u$, so $w(C_u)=w(C_v)$.

Suppose now that $C_v$ is already visited when $C_u$ is processed. Let
$d'=d(C_v)$. If $C_v$ was processed in an earlier wave and $d'=-d$, then $C_u$
was a forward neighbor of $C_v$, so rule~\textup{(F)} would already have oriented
the edge from $C_v$ into $C_u$, a contradiction. If $C_v$ was processed in an
earlier wave and $d'=d$, then $C_u$ was an unvisited inward neighbor of $C_v$ and
would have been processed no later than the next wave, before any later wave in
direction $d$, again a contradiction. Hence $C_v$ must have been processed earlier
in the same wave as $C_u$, and therefore $w(C_u)=w(C_v)$.

Since both components are processed in this wave, $d(C_u) = d(C_v) = d$. Moreover, $C_v$ lies on the forward side of $C_u$, so
$x(v)=x(u)+d$, which implies $d\,x(v)=d\,x(u)+1$. Thus the second coordinate of $\Phi$
strictly increases, so again $\Phi(u)<\Phi(v)$.
This proves~(ii).
\end{proof}

\noindent
We now verify the three defining properties of a rooted pinch-free orientation.

\begin{lem}\label{lem:P1}
The orientation $\overrightarrow{G}$ satisfies property \textup{(P1)}:
it is pinch-free.
\end{lem}

\begin{proof}
Fix a vertex $v$, and let $C_v$ be the column component containing $v$.

\medskip
\noindent\emph{Vertical in-degree.}
When $C_v$ is processed, rule~\textup{(V)} orients all vertical edges of $C_v$ away from
$s(C_v)$. Since $C_v$ is vertically connected,
every vertex of $C_v$ except $s(C_v)$ has exactly one incoming vertical edge. Thus $v$ has
vertical in-degree at most~$1$.

\medskip
\noindent\emph{Horizontal in-degree.}
When $C_v$ is processed, rule~\textup{(F)} orients every horizontal edge from $C_v$ to
its forward side out of $C_v$. By~\cref{lem:oriented-once}, these edges are oriented only 
once. Hence $v$ has no incoming horizontal edge from the forward side of $C_v$, so any
incoming horizontal edge at $v$ must come from the inward side. Since $v$ has at most
one horizontal neighbor on that side, $v$ has horizontal in-degree at most~$1$.

Therefore $v$ cannot have two collinear incoming edges. Hence no vertex of
$\overrightarrow{G}$ is pinched.
\end{proof}

\medskip
\noindent
To prove acyclicity, we use the lexicographically ordered potential function introduced
in \cref{def:potential}. 

\begin{lem}\label{lem:P2}
The orientation $\overrightarrow{G}$ satisfies property \textup{(P2)}: it is acyclic.
\end{lem}

\begin{proof}
By~\cref{lem:potential}, the potential $\Phi$ strictly increases along every
directed edge of $\overrightarrow{G}$. A directed cycle would therefore yield a strict
increase all the way around the cycle, which is impossible. Hence $\overrightarrow{G}$
contains no directed cycle.
\end{proof}

\begin{lem}\label{lem:P3}
The orientation $\overrightarrow{G}$ satisfies property \textup{(P3)}: it is rooted.
\end{lem}

\begin{proof}
Let $C_r$ be the root component chosen in Algorithm~\ref{alg:wavefront}, and let
$r=s(C_r)$ be its seed. We prove by induction on the processing order of the
column components that every vertex of each processed component is reachable from $r$.

For the base case, the first processed component is $C_r$. Rule~\textup{(V)}
orients all vertical edges of $C_r$ away from $r$, so every vertex of $C_r$ is
reachable from $r$.

Now let $H\neq C_r$ be a newly processed component, and assume the claim holds for
all previously processed components. The seed $s(H)$ is chosen
to be a vertex of $H$ adjacent to a vertex $u$ in some already visited inward
neighbor component $A$.

We claim that the edge between $u$ and $s(H)$ is oriented as $u\to s(H)$.
When $A$ is processed, the component $H$ is still unvisited. Hence rule~\textup{(I)}
cannot apply to the pair ${A,H}$. Therefore the edge between $u$ and $s(H)$ can be
oriented at that time only if $H$ lies on the forward side of $A$, in which case
rule~\textup{(F)} orients it as $u\to s(H)$.
Otherwise, the edge is still unoriented when $H$ is processed.
At that moment, $A$ is a visited inward neighbor of $H$, so rule~\textup{(I)}
orients the horizontal edges between $A$ and $H$ from $A$ into $H$, again giving
$u\to s(H)$.

By the induction hypothesis, $u$ is reachable from $r$, and therefore so is
$s(H)$. Rule~\textup{(V)} then orients all vertical edges of $H$ away from
$s(H)$, so every vertex of $H$ is reachable from $s(H)$, and hence from $r$.
Thus every vertex of every processed component is reachable from $r$. Since every
component is eventually processed by~\cref{prop:W-bound}, every vertex of $G$ is
reachable from $r$. Therefore $\overrightarrow{G}$ is rooted.
\end{proof}

\begin{theorem}[Wavefront correctness]
\label{thm:wavefront-correctness}
$\overrightarrow{G}$ is a rooted pinch-free orientation of $G$.
\end{theorem}

\begin{proof}
By \cref{lem:P1,lem:P2,lem:P3}, the orientation satisfies properties
\textup{(P1)}--\textup{(P3)} of Definition~\ref{def:roo}.
Thus it is pinch-free, acyclic, and rooted.
\end{proof}

\subsection{Complexity of Wavefront Sweep}\label{sec:complexity}

Let $n$ be the number of cells of the input polyomino $P$.
We assume that the input is given as a list of
occupied grid coordinates sorted lexicographically by $(x,y)$. The same bounds hold if the input is given as a $2$D occupancy array whose bounding box has size $O(n)$. 

The algorithm has two stages: preprocessing, which computes the column
decomposition and the horizontal adjacencies between column components, and the
wavefront phase, which orients the edges.

\smallskip\noindent
\emph{Preprocessing.}
We identify maximal vertical runs and their
horizontal adjacencies by scanning cells column-by-column and inspecting 
horizontal neighbor pairs across consecutive columns. 
Since the dual graph $G$ has maximum degree~4, it has $O(n)$ edges.
Since each cell and each grid edge is examined only $O(1)$ times, preprocessing takes
$O(n)$ time and $O(n)$ auxiliary space.

\smallskip\noindent
\emph{Wavefront phase.}
Each column component is processed exactly once, and each edge of $G$ is 
oriented exactly once by \cref{prop:W-bound,lem:oriented-once}. 
All operations use only constant-time work per component or per edge,
assuming standard $O(1)$ worklist operations for $W$ and the frontiers
$\Gamma_{\pm1}$. Hence the wavefront phase also runs in $O(n)$ time using
$O(n)$ auxiliary space. This establishes~\cref{thm:main}. 
 
 \section{Conclusion and Future Work}
\label{sec:conclusion}

We gave a linear-time algorithm that computes a rooted pinch-free orientation of
the dual graph of any connected polyomino. Thus, in the unrestricted-access model,
where robots can access any placement site, for example by flying, this
orientation provides a complete construction framework. 

\begin{figure}[!htp]
    \centering
    \includegraphics[page=3,width=0.7\textwidth]{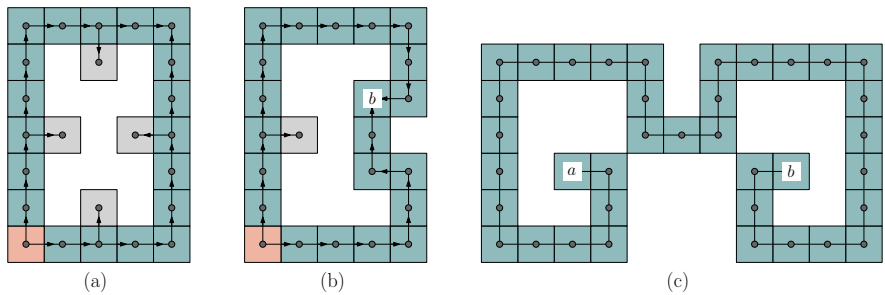}
    \caption{Outside-access limitations of rooted pinch-free orientations. (a) Closing the outer cycle too early blocks access to the interior. (b) The polyomino is buildable, but no topological order of the shown wavefront orientation is buildable, since $b$ must be placed last. (c) A non-buildable polyomino whose dual graph is a path; any rooted pinch-free orientation would require a root vertex with directed paths to both $a$ and $b$; some root choices on the perimeter of the left enclosed region can make $a$ reachable, and similarly for $b$ on the right; however, no root choice can make both $a$ and $b$ reachable.}
        \label{fig:seal}
\end{figure}

The stronger setting in which each cell to be placed must be reachable from the 
unbounded free region of the current partial construction is not captured by the orientation alone. 
Figure~\ref{fig:seal}
illustrates three distinct situations. In~\cref{fig:seal}(a), the wavefront
orientation admits a topological order that is buildable, but not every
topological order is: some topological orders complete the outer cycle before the
interior is built and thereby block access to the remaining cells. In
\cref{fig:seal}(b), the polyomino itself is buildable, but no topological order
of that particular wavefront orientation is buildable. Thus the orientation
correctly captures attachment dependencies and pinch-freeness, but it does not in
general encode outside-access constraints.~\cref{fig:seal}(c) shows
that some polyominoes are not buildable at all under the outside-access model.

These examples show that reachability from the unbounded free space is a global constraint not captured by rooted pinch-free orientations. Thus, wavefront orientation is only one part of construction planning. Characterizing which polyominoes admit such outside-access build orders under the local attachment rules considered here, and computing them efficiently, remain interesting directions for future work.   
\bibliographystyle{abbrv}
\bibliography{assembly}

\end{document}